\documentclass{amsart}
\usepackage[english]{babel}
\usepackage[latin1]{inputenc}
\usepackage[dvips,final]{graphics}
\usepackage{amsmath,amsfonts,amssymb,amsthm,amscd,array,
stmaryrd,mathrsfs,mathdots,epigraph}
\usepackage[makeroom]{cancel}
\usepackage{pstricks}
\usepackage[all]{xy}
\usepackage{url}
\usepackage{multirow, blkarray}
\usepackage{booktabs}
\usepackage{caption}
\usepackage{subcaption}

\usepackage{blindtext}
\usepackage{textcomp}
\usepackage[final]{epsfig}
\usepackage{color}
\usepackage{mathabx}

\theoremstyle{plain}

\newtheorem{lem}{Lemma}[section]
\newtheorem{thm}[lem]{Theorem}
\newtheorem{cor}[lem]{Corollary}
\newtheorem{prop}[lem]{Proposition}

\theoremstyle{definition}
\newtheorem{rem}[lem]{Remark}
\newtheorem{ex}[lem]{Example}

\newtheorem{defn}[lem]{Definition}

\numberwithin{figure}{section}
\numberwithin{table}{section}

\let\ssection=\section
\renewcommand{\section}{\setcounter{equation}{0}\ssection}

\newcommand{\R}{\mathbb{R}}
\newcommand{\Z}{\mathbb{Z}}

\newcommand{\bP}{\mathbb{P}}
\newcommand{\Q}{\mathbb{Q}}

\newcommand{\E}{\mathcal{E}}
\newcommand{\F}{\mathcal{F}}

\newcommand{\cK}{\mathcal{K}}
\newcommand{\cL}{\mathcal{L}}
\newcommand{\cR}{\mathcal{R}}

\newcommand{\cU}{\mathcal{U}}
\newcommand{\cV}{\mathcal{V}}

\newcommand{\SL}{\mathrm{SL}}
\newcommand{\PSL}{\mathrm{PSL}}
\newcommand{\OSp}{\mathrm{OSp}}

\newcommand{\ES}{\mathrm{ES}}
\newcommand{\FS}{\mathrm{FS}}
\newcommand{\OS}{\mathrm{OS}}
\newcommand{\rmS}{\mathrm{S}}

\def\a{\alpha}
\def\b{\beta}
\def\d{\delta}

\def\g{\gamma}

\def\l{\lambda}

\def\thup{\mathop{\rm th}\nolimits}

\begin{document}

\title[Shadows of rationals and irrationals]
{Shadows of rationals and irrationals:\\
supersymmetric continued fractions\\
and the super modular group}

\author{Charles H.\ Conley}
\address{
Charles H.\ Conley,
Department of Mathematics 
\\University of North Texas 
\\Denton TX 76203, USA} 
\email{conley@unt.edu}

\author{Valentin Ovsienko}
\address{
Valentin Ovsienko,
Centre National de la Recherche Scientifique,
Laboratoire de Math\'ematiques de Reims, UMR9008 CNRS,
Universit\'e de Reims Champagne-Ardenne,
U.F.R. Sciences Exactes et Naturelles,
Moulin de la Housse - BP 1039,
51687 Reims cedex 2,
France}
\email{valentin.ovsienko@univ-reims.fr}

\thanks{\noindent
C.H.C.\ was partially supported by Simons Foundation Collaboration Grant~519533.\\
\indent V.O.\ was partially supported by the ANR project PhyMath, ANR-19-CE40-0021.
}



\begin{abstract}
This paper is an attempt to apply the tools of supergeometry to arithmetic.
Supergeometric objects are defined over
supercommutative rings of coefficients, and
we consider an integral ring with exactly two odd variables.
In this case the even quantities, such as numbers and continued fractions,
are ``doubled'', having both a classical and a nilpotent part.
We refer to the nilpotent part as the ``shadow''.
We investigate the notions of supersymmetric continued fractions
and the orthosymplectic modular group
and make some initial steps toward studying their properties.
\end{abstract}

\maketitle

\thispagestyle{empty}

\tableofcontents

\section{Introduction}\label{Intro}

In supergeometry, Lie supergroups and other geometric objects are considered over 
a $\Z_2$-graded ring $R=R_{\bar0}\oplus R_{\bar1}$ of coefficients.
This ring, however, usually remains unspecified.
This is perhaps the reason for which supergeometry
rarely produces concrete numeric sequences.
(There are exceptions.  For example, in \cite{Fre}, pp.~21--22, the ring
$\R[\eta_1,\ldots,\eta_L]$ with finitely many ``auxiliary odd parameters'' was considered.)

The idea of the ``shadow'' of a number \cite{Ovs1}
arose from supergeometry and cluster superalgebras
(cf.~\cite{Ovs, OS, OT}).
Suppose that the odd part of~$R$ contains
exactly two generators, say $\xi$ and~$\eta$.
One example of such a ring
is the following superextension of the integers,
the ring of coordinates on $\Z^{1|2}$:
\begin{equation}
\label{REq}
R=(\Z\oplus\Z\xi\eta)\oplus(\Z\xi\oplus\Z\eta).
\end{equation}
This is the ``minimal'' choice of the ring of coefficients in which
an even variable $A$ may have a non-trivial nilpotent part: $A=a+a'\xi\eta$.
In this situation, any natural procedure or algorithm acting on an integer $a$
will produce a sequence or quantity having not only a classical part,
but also an even nilpotent part, the shadow.
The notion of shadows was recently tested in~\cite{Hon,Ves},
where the idea was applied to sequences of integers
in the context of algebraic geometry and number theory.

The main goal of this paper is to introduce a notion of
supersymmetric continued fractions, together with the corresponding Farey tree.
This allows one to calculate the shadows of rational numbers.
The property of convergence of continued fractions
then extends the definition of shadows to irrationals.
It is amusing to note that a rational number may have multiple shadows,
while an irrational number seems to have only one.
It is difficult to say at this stage if these ``super'' continued fractions will have applications,
but their properties are quite nice and the explicit formulas are harmonious.

The second goal of this paper, closely related to the first,
is to study the notion of the ``super modular group'',
by which we mean the supergroup $\OSp(1|2)$
with coefficients in the ring~\eqref{REq}.
We denote this group by $\OSp(1|2, \Z)$.

Let us note that the recent work \cite{MOZ1} exploits related ideas
applied to hyperbolic geometry and combinatorics.
We were not aware of this reference while working on the present paper.

\subsection{Continued fractions}\label{IntroCF}
The regular finite continued fraction $[a_1,\ldots,a_n]$ is the expression
\begin{equation}
\label{CEq}
[a_1,\ldots,a_n]:=
a_1 + \cfrac{1}{a_2 + \cfrac{1}{\ddots +\cfrac{1}{a_{n}} } } \,,
\end{equation}
where $a_i \in \Z_{>0}$ for all $i > 1$.
Every rational number has exactly two finite continued fraction expansions.
This is due to the ambiguity $[a_1,\ldots,a_n]=[a_1,\ldots,a_n-1,\,1]$
for $a_n\geq2$: the \textit{length}~$n$ in~\eqref{CEq}
may be taken to be either even or odd.

Consider the well-known triangular generators of the modular group $\SL(2,\Z)$:
$$
R=\begin{pmatrix}
1&1\\[2pt]
0&1
\end{pmatrix},
\qquad\qquad
L=\begin{pmatrix}
1&0\\[2pt]
1&1
\end{pmatrix}.
$$
The continued fraction~\eqref{CEq} corresponds to the word
$R^{a_1} L^{a_2} R^{a_3} \cdots$ in $R$ and $L$ in the following sense.
Suppose that $\frac{p}{q}$ is a rational number in reduced form,
and $[a_1,\ldots,a_n]$ is one of its two continued fraction expansions.
Then if $n$ is even, say $n = 2m$, we have
\begin{equation}
\label{Prod1Eq}
\begin{pmatrix}
p\\[2pt]
q
\end{pmatrix}
=
R^{a_1}L^{a_2}\cdots{}R^{a_{2m-1}}L^{a_{2m}}
\begin{pmatrix}
1\\[2pt]
0
\end{pmatrix},
\end{equation}
while if $n$ is odd, say $n = 2m+1$, we have
\begin{equation}
\label{Prod2Eq}
\begin{pmatrix}
p\\[2pt]
q
\end{pmatrix}
=
R^{a_1}L^{a_2}\cdots{}L^{a_{2m}}R^{a_{2m+1}}
\begin{pmatrix}
0\\[2pt]
1
\end{pmatrix}.
\end{equation}

Now consider the $3\times3$ matrices
\begin{equation}
\label{RSEq}
\cR=\begin{pmatrix}
1&\phantom{-}1&\phantom{-}\xi\\[2pt]
0&\phantom{-}1&\phantom{-}0\\[2pt]
0&\phantom{-}\xi&\phantom{-}1
\end{pmatrix},
\qquad\qquad
\cL=\begin{pmatrix}
1&\phantom{-}0&\phantom{-}0\\[2pt]
1&\phantom{-}1&-\eta\\[2pt]
\eta&\phantom{-}0&\phantom{-}1
\end{pmatrix},
\end{equation}
where $\xi$ and $\eta$ are Grassmann variables, i.e.,
$$\xi^2=\xi\eta+\eta\xi=\eta^2=0.$$
These matrices belong to $\OSp(1|2, \Z)$, and
they seem to be the most natural superanalogues of $R$ and $L$.
They have frequently appeared in the literature; for instance,
$\cR$ was understood as the translation operator in \cite{Rab}.

In our definition of supersymmetric continued fractions,
we replace the matrices $R$ and $L$ in \eqref{Prod1Eq}
and \eqref{Prod2Eq} by $\cR$ and $\cL$,
apply the resulting words in $\cR$ and $\cL$
to the vectors $(1,0,0)^t$ and $(0,1,0)^t$, respectively,
and take the quotient of the two even coordinates.
This gives what we call the (finite) ``supersymmetric continued fraction'',
$$
\{a_1,\ldots,a_{n}\}=
[a_1,\ldots,a_n]+[a_1,\ldots,a_n]_\rmS\, \xi\eta.
$$
We refer to the coefficient in the nilpotent part,
$[a_1,\ldots,a_n]_\rmS$, as the shadow of $[a_1,\ldots,a_n]$.
Our main result is the following convergence property, which allows us to extend
the notion of shadows to irrationals.

\begin{thm}
\label{ConvThm}
For any integer sequence $a_1, a_2, a_3, \ldots$
such that $a_1 \geq 0$ and $a_i \geq 1$ for $i>1$,
the rational sequence $[a_1,a_2,a_3,\ldots,a_n]_\rmS$ converges.
Furthermore, the limit is positive if and only if $a_1 \geq 1$.
\end{thm}

Let $x$ be an irrational number with
continued fraction expansion $[a_1,a_2,a_3,\ldots]$.
We define the \textit{shadow} $(x)_\rmS$ of $x$ to be
$$
\left(x\right)_\rmS:=\lim_{n\to \infty} \left[a_1,a_2,a_3,\ldots,a_n\right]_\rmS.
$$
Theorem~\ref{ConvThm} is proven in Section~\ref{ConvSec}.
Its proof, like its statement, closely parallels the classical case.

\subsection{Two shadows of rationals}\label{IntroShR}
Consider a rational number with reduced expression $\frac{p}{q}$.
As noted earlier, it has two continued fraction expansions:
\begin{equation} \label{classical}
\frac{p}{q}=[a_1,\ldots,a_n]=[a_1, \ldots, a_{n-1}-1,\,1].
\end{equation}
The corresponding supersymmetric continued fractions give two different shadows,
$[a_1,\ldots,a_n]_{\rmS}$ and $[a_1, \ldots, a_{n-1}-1,\,1]_{\rmS}$,
which we call the \textit{even} and \textit{odd} shadows
of $\frac{p}{q}$, according to their length.
We denote them by $\big(\frac{p}{q}\big)_{\ES}$
and $\big(\frac{p}{q}\big)_{\OS}$.
Both are rational numbers, but they have quite different properties; 
see Sections~\ref{PropTransSec} and~\ref{LocSec}.

One heuristic property the even and odd shadows appear
to have in common is a certain ``fractal nature''
of the functions $\ES:\Q\to\Q$ and $\OS:\Q\to\Q$.
Computer experiments suggest that they are
neither continuous nor monotonic.
Some examples are given in Section~\ref{SmallExSec}.
In view of these observations, the convergence of
Theorem~\ref{ConvThm} seems particularly surprising.

It may appear strange that rational numbers have several different shadows,
while irrationals seem to have only one.  In fact, a
similar phenomenon occurs for $q$-deformations,
as noticed in~\cite{SVRe} and explained in~\cite{Lic}.

\subsection{Continuant polynomials}\label{IntroCP}

Classical continued fractions are related to certain remarkable polynomials
in several variables, $K_n(a_1,\ldots,a_n)$, known as \textit{continuants}.
Indeed, taking $a_1,\ldots,a_n$ in~\eqref{CEq} as variables, 
the continued fraction is given by the quotient of two continuants; 
see~\cite{Concr, BR1} and Section~\ref{EKSec}.

Supersymmetric continued fractions
may also be written as quotients of polynomials:
$$
\left\{a_1,\ldots,a_n\right\}=
\frac{\cK_n(a_1,\ldots,a_n)}{\cK_{n-1}(a_2,\ldots,a_n)},
$$
where $\cK_n(a_1,\ldots,a_n)=K_n(a_1,\ldots,a_n)+K'_n(a_1,\ldots,a_n)\xi\eta$.
The ``shadow'' part, $K'_n$, is a weighted version of $K_n$.
In Section~\ref{ExpCon} we derive its explicit formula.
This allows us to deduce several properties of 
supersymmetric continued fractions, including our convergence result.

The shadowed continuants connect the present paper to~\cite{MGOT, Ust}.
Indeed, the polynomials $\cK_n$ are special cases of the supercontinuants
arising from the notion of frieze patterns.

\subsection{The Farey tree}\label{FareyIntro}

The Farey (or Stern-Brocot) tree is a beautiful way to visualize the set of rational numbers,
completed with infinity, represented by $\frac{1}{0}$.
Each rational appears exactly once, labelling the connected components
of the planar complement of the tree.
In addition to continued fractions, the Farey tree is related to
hyperbolic geometry and many other subjects.

Our definition of the super (or ``shadowed'') Farey tree
begins with the following initial ``fishbone'' diagram:
\begin{equation}
\label{SuperFish}
\xymatrix @!0 @R=0.38cm @C=0.9cm
{
&\mkern-14mu -1&&\\
&\mkern-14mu \phantom{-}1&&\\
&\mkern-14mu \phantom{-}\eta&&\\
\\
1&\bullet\ar@{-}[ruu]\ar@{-}[luu]\ar@{-}[dd]&0\\
0&&1\\
0&\bullet\ar@{-}[ldd]\ar@{-}[rdd]&0\\
\\
&1&&\\
&1&&\\
&\xi&&
}
\end{equation}
This is justified by the fact that $\OSp(1|2)$ acts $3|2$-transitively
on the projective line $\bP^{1|1}$; see~\cite{DM}.
Therefore any three points of~$\bP^{1|1}$ may be sent to the set
$$
\big\{ (1:0:0),\ (0:1:0),\ (1:1:\xi) \big\},
$$
where $\xi$ an odd, or Grassmann, parameter.
We introduce the second parameter $\eta$ in a symmetric manner.

It turns out that the diagram~\eqref{SuperFish}
has a symmetry described by the matrices
\begin{equation}
\label{UVEq}
\cU=\begin{pmatrix}
\phantom{-}0&\phantom{-}1&\phantom{-}0\\[2pt]
-1&-1&\phantom{-}\eta\\[2pt]
\phantom{-}0&-\eta&\phantom{-}1
\end{pmatrix},
\qquad\qquad
\cV=\begin{pmatrix}
-1&1&\xi\\[2pt]
-1&0&0\\[2pt]
-\xi&0&1
\end{pmatrix}.
\end{equation}
They are related to $\cR$ and $\cL$ via
\begin{equation}
\label{RLUVEq}
\cV=\cR S,
\qquad
\cU=\cL^{-1}S,
\qquad
\hbox{where}
\qquad
S=\begin{pmatrix}
\phantom{-}0&1&0\\[2pt]
-1&0&0\\[2pt]
\phantom{-}0&0&1
\end{pmatrix}.
\end{equation}
Let us mention that the matrices of this form were understood in~\cite{MGOT} 
as superanalogues of the discrete Sturm-Liouville operator.

\section{Supersymmetric continued fractions and their simplest properties}\label{SSContF}

In this section we give a more detailed definition of
the supersymmetric, or ``shadowed'', continued fraction.
We then define the even and odd shadows of a rational number
and describe some of their general properties.

\subsection{Definitions and notations}\label{NotSec}

Consider a classical continued fraction, $\left[a_1,\ldots,a_{n}\right]=\frac{p}{q}$,
where as usual the $a_i$ are positive integers for $i > 1$
and $\frac{p}{q}$ is given in reduced terms.
We define the supersymmetric continued fraction 
$\left\{a_1,\ldots,a_{n}\right\}$ in terms of the matrices
$\cR$ and $\cL$ given in \eqref{RSEq}:

\begin{defn}
\label{TheMainDefn}
Define integers $p'$, $q'$, $\lambda$, and $\mu$ via the following equations.
For $n=2m$ even, set
\begin{equation}
\label{EvenProd}
\cR^{a_1}\cL^{a_2}\cdots\cR^{a_{2m-1}}\cL^{a_{2m}}
\begin{pmatrix}
1\\
0\\
0
\end{pmatrix}=
\begin{pmatrix}
p+p'\xi\eta\\
q+q'\xi\eta\\
\l\xi+\mu\eta
\end{pmatrix}.
\end{equation}
For $n=2m+1$ odd, set
\begin{equation}
\label{OddProd}
\cR^{a_1}\cL^{a_2}\cdots\cL^{a_{2m}}\cR^{a_{2m+1}}
\begin{pmatrix}
0\\
1\\
0
\end{pmatrix}=
\begin{pmatrix}
p+p'\xi\eta\\
q+q'\xi\eta\\
\l\xi+\mu\eta
\end{pmatrix}.
\end{equation}

For all $n$, set
\begin{equation}
\label{SCFEq}
\left\{a_1,\ldots,a_{n}\right\}:=
\frac{p+p'\xi\eta}{q+p'\xi\eta}=
\frac{p}{q}+\frac{p'q-pq'}{q^2}\xi\eta.
\end{equation}
The coefficient of $\xi\eta$ is then the shadow
of the supersymmetric continued fraction:
\begin{equation}
\label{QuotShad}
\left[a_1,\ldots,a_{n}\right]_{S}:=\frac{p'q-pq'}{q^2}.
\end{equation}
\end{defn}

We also define the even and odd shadows of the rational $\frac{p}{q}$, written as
$$
\Big(\frac{p}{q}\Big)_{\ES},
\qquad\qquad
\Big(\frac{p}{q}\Big)_{\OS}.
$$
They are the shadows of the even and odd supersymmetric continued fractions
coming from \eqref{classical}: one of them is $[a_1, \ldots, a_n]_\rmS$,
and the other is $[a_1, \ldots, a_{n-1}-1,\, 1]_\rmS$,
which is which being determined by the parity of $n$.
In general, they are different.

\begin{ex}
\label{SimpEx}
To give a simple example, consider
$\frac{5}{2}=[2,2]=[2,1,1]$.
One finds that
$$
\cR^2\cL^2
\begin{pmatrix}
1\\
0\\
0
\end{pmatrix}=
\begin{pmatrix}
5+4\xi\eta\\
2\\
4\xi+2\eta
\end{pmatrix},
\qquad\qquad
\cR^2\cL\,\cR
\begin{pmatrix}
0\\
1\\
0
\end{pmatrix}=
\begin{pmatrix}
5+4\xi\eta\\
2+\xi\eta\\
5\xi+\eta
\end{pmatrix}.
$$
The two quotients
$$
\{2,2\}=\frac{5+4\xi\eta}{2}=\frac{5}{2}+2\xi\eta,
\qquad\qquad
\{2,1,1\}=\frac{5+4\xi\eta}{2+\xi\eta}=\frac{5}{2}+\frac{3}{4}\xi\eta
$$
give the shadows $\big(\frac{5}{2}\big)_{\ES}=2$
and $\big(\frac{5}{2}\big)_{\OS}=\frac{3}{4}$.
\end{ex}

\subsection{Small examples}\label{SmallExSec}
The even and odd shadows of all rationals $1\leq\frac{p}{q}\leq2$ with 
denominators $q\leq12$ are depicted in Figs.~\ref{Sha12Fig} and~\ref{OddSha12}, respectively.
Note that the values seem to be increasingly ``fractal'' as $q$ grows.
\begin{figure}[htpb]
    \centering
    \includegraphics[width=0.7\textwidth]{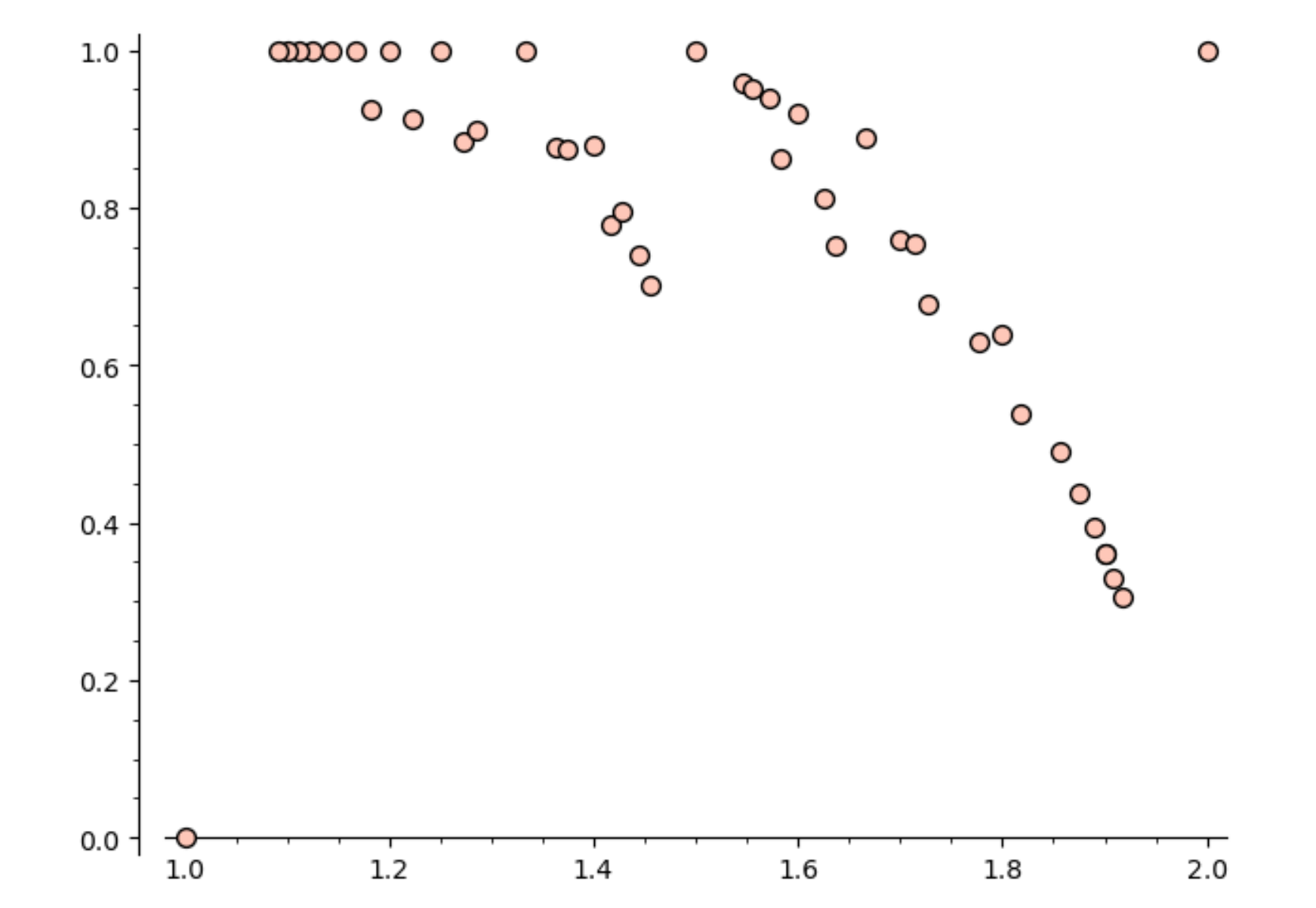}
    \caption{Even shadows of small rationals.}
    \label{Sha12Fig}
\end{figure}
\begin{figure}[htpb]
    \centering
    \includegraphics[width=0.7\textwidth]{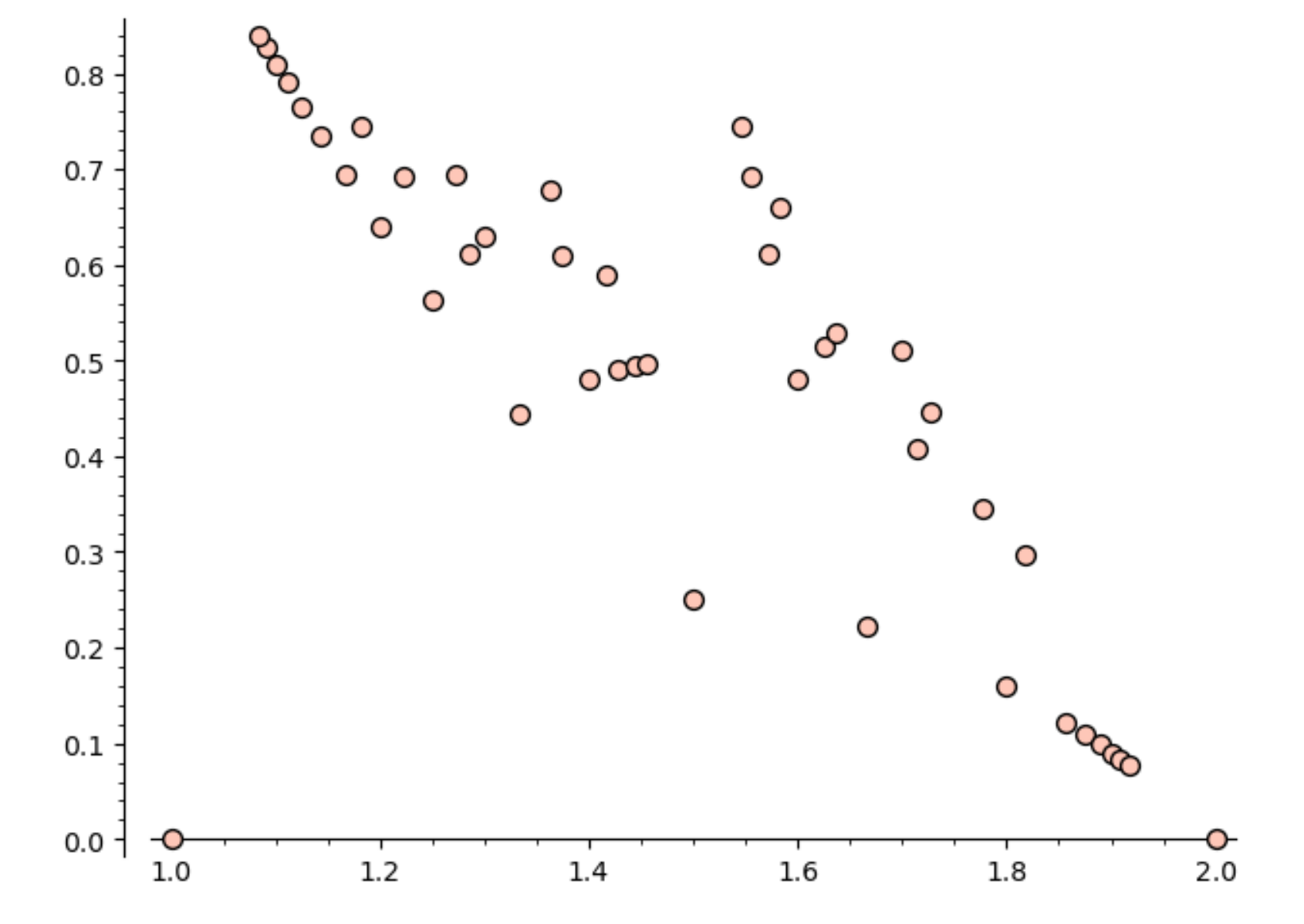}
    \caption{Odd shadows of small rationals.}
    \label{OddSha12}
\end{figure}

\subsection{Translation invariance and shadows of integers}\label{PropTransSec}
The even shadow is translation-invariant, but the odd shadow is not.

\begin{prop}
\label{TransProp}
One has 
$$
\Big(\frac{p}{q}+1\Big)_{\ES}=
\Big(\frac{p}{q}\Big)_{\ES}+1.
$$
\end{prop}

\begin{proof}
Let us calculate how the operator $\cR$ changes the shadow of a rational.
Apply it to an arbitrary vector:
\begin{equation} \label{temp}
\begin{pmatrix}
1&1&\xi\\[2pt]
0&1&0\\[2pt]
0&\xi&1
\end{pmatrix}
\begin{pmatrix}
p+p'\xi\eta\\[2pt]
q+q'\xi\eta\\[2pt]
\l\xi+\mu\eta
\end{pmatrix}=
\begin{pmatrix}
p+q+(p'+q'+\mu)\xi\eta\\[2pt]
q+q'\xi\eta\\[2pt]
(\l+q)\xi+\mu\eta
\end{pmatrix}.
\end{equation}
Recall the quotient \eqref{SCFEq} defining the supersymmetric continued fraction.
The quotient corresponding to the right hand side of \eqref{temp} is
$$
\frac{p+q}{q}+\frac{(p+q)q'-(p'+q'+\mu)q}{q^2}\xi\eta=
\frac{p+q}{q}+\frac{p'q-pq'}{q^2}\xi\eta+\frac{\mu}{q}\xi\eta.
$$
Thus $\cR$ changes the shadow $\big(\frac{p}{q}\big)_{\ES}$ by adding $\frac{\mu}{q}$.
The equation \eqref{temp} also shows that $\cR$ does not change $q$ or $\mu$.

Now consider the action of $\cL$.  Check that if $q$ and $\mu$ are equal
before applying $\cL$, they remain equal afterward:
$$
\begin{pmatrix}
1&0&\!0\\[2pt]
1&1&\!-\eta\\[2pt]
\eta&0&\!1
\end{pmatrix}
\begin{pmatrix}
p+p'\xi\eta\\[2pt]
q+q'\xi\eta\\[2pt]
\l\xi+\mu\eta
\end{pmatrix}=
\begin{pmatrix}
p+p'\xi\eta\\[2pt]
p+q+(p'+q'+\l)\xi\eta\\[2pt]
\l\xi+(p+\mu)\eta
\end{pmatrix}.
$$
It follows by induction that in all the partial products of~\eqref{EvenProd},
the coefficients $\mu$ and $q$ are equal.
In particular, the even supersymmetric continued fractions
$\{a_1,a_2,\ldots,a_{2m}\}$ and $\{a_1+1,\,a_2,\ldots,a_{2m}\}$
indeed differ by~$1$.
\end{proof}

Coupling Proposition~\ref{TransProp} with
Theorem~\ref{ConvThm} gives the following corollary.

\begin{cor} \label{irrat trans}
If $x$ is any positive irrational, then $(x+1)_\rmS=(x)_\rmS+1$.
\end{cor}

Translation-invariance does not hold for odd shadows.
For instance, $\big(\frac{5}{2}\big)_{\OS}=\frac{3}{4}$,
but $\big(\frac{7}{2}\big)_{\OS}=\frac{5}{4}$.
As regards integers, a short computation gives
\begin{equation} \label{int case}
(n)_{\ES}=n-1,
\qquad\qquad
(n)_{\OS}=0.
\end{equation}

\subsection{Accordance between even and odd continued fractions}
\label{EvenOddSec}

Suppose that $[a_1,a_2,\ldots,a_{2m}]$ and $[a'_1,a'_2,\ldots,a'_{2m\pm1}]$
are even and odd continued fractions representing the same rational.
Under our choice of the initial vectors in Definition~\ref{TheMainDefn},
the corresponding even and odd supersymmetric continued fractions may differ.
However, as demonstrated by the following proposition,
there is a different choice of the initial vectors under which
the even and odd supersymmetric continued fractions are the same.

\begin{prop}
\label{AccordProp}
If\/ $[a_1,a_2,\ldots,a_{2m}]=[a'_1,a'_2,\ldots,a'_{2m\pm1}]$, then
$$
\cR^{a_1}\cL^{a_2}\cdots\cR^{a_{2m-1}}\cL^{a_{2m}}
\begin{pmatrix}
\phantom{-}1\\
\phantom{-}0\\
-\eta
\end{pmatrix}=
\cR^{a'_1}\cL^{a'_2}\cdots\cL^{a'_{2m}}\cR^{a'_{2m\pm1}}
\begin{pmatrix}
\phantom{-}0\\
\phantom{-}1\\
-\xi
\end{pmatrix}.
$$
\end{prop}

The proof reduces to the observation that
$\cL(e_1 - \eta e_3)$ and $\cR(e_2 - \xi e_3)$ are equal:
both are $e_1 + e_2$.
(Here we write $e_i$ for the standard basis vectors.)
But although the argument is simple,
we do not yet have a conceptual understanding of the result.
It may be that these new initial vectors and the associated
accordance of the even and odd supersymmetric continued fractions
have a deeper significance.

\section{Analytic formulas for continued fractions}\label{ExpSec}

Our next goal is to give an explicit formula for supersymmetric continued fractions.
As in the classical case, the main ingredient is the continuant polynomial.
We calculate the ``shadowed continuant'' arising from supersymmetric continued fractions.

\subsection{Euler's continuants}\label{EKSec}

Recall that classical continued fractions \eqref{CEq} may be expressed
in terms of the \textit{continuant} polynomial $K_n(a_1,\ldots,a_n)$,
defined as the $n\times{}n$ determinant
\begin{equation} 
\label{ContEq}
K_n(a_1,\ldots,a_n):=
\det
\begin{pmatrix}
a_1&1&&&\\[4pt]
\!-1&a_{2}&1&&\\[4pt]
&\ddots&\ddots&\!\!\ddots&\\[4pt]
&&-1&a_{n-1}&\!\!\!\!\!1\\[4pt]
&&&\!\!\!\!-1&\!\!\!\!a_{n}
\end{pmatrix}.
\end{equation}
The first few continuants are
\begin{eqnarray*}
K_0&=&1,
\\
K_1(a_1)&=&a_1,
\\
K_2(a_1,a_2)&=&1+a_1a_2,
\\
K_3(a_1,a_2,a_3)&=&a_1+a_3+a_1a_2a_3,
\\
K_4(a_1,a_2,a_3,a_4)&=&1+a_1a_2+a_1a_4+a_3a_4+a_1a_2a_3a_4.
\end{eqnarray*}

They satisfy the recurrence formula
\begin{equation}
\label{ClassRec}
K_n(a_1,\ldots,a_n)=a_1K_{n-1}(a_2,\ldots,a_n)+K_{n-2}(a_3,\ldots,a_n),
\end{equation}
which is equivalent to the fact that their quotients
are the classical continued fractions:
\begin{equation} 
\label{ContCFEq}
\left[a_1,\ldots,a_n\right]=
\frac{K_n(a_1,\ldots,a_n)}{K_{n-1}(a_2,\ldots,a_n)}.
\end{equation}
Continuants have been studied since the time of Euler~\cite{Eul}
and have many beautiful properties \cite{Concr, BR1}.

\subsection{The shadowed continuant}\label{SKSec}

\begin{defn}
\label{ShContDef}
The \textit{shadowed continuant} $\cK_n$ is the polynomial in~$n$ commuting variables
$a_1,\ldots,a_n$ and two Grassmann variables $\xi, \eta$
given by
\begin{equation}
\label{ShContEq}
\cK_n(a_1,\ldots,a_n)=
\left\{
\begin{array}{lr}
\big(1+\xi\eta(\E-1)\big) K_n(a_1,\ldots,a_n) + \xi\eta,
& n\;\hbox{even},\\[6pt]
\big(1+\xi\eta(\E-1)\big) K_n(a_1,\ldots,a_n),
& n\;\hbox{odd}.
\end{array}
\right.
\end{equation}
Here $\E$ is the Euler operator:
\begin{equation} 
\label{EulerEq}
\sum_i\E=a_i\frac{\partial}{\partial a_i}.
\end{equation}
\end{defn}

Equivalently, 
\begin{equation} \label{ShContDefBis}
\cK_n(a_1,\ldots,a_n)=K_n(a_1,\ldots,a_n)+\xi\eta K'_n(a_1,\ldots,a_n),
\end{equation}
where the ``shadow part'' is
\begin{equation} 
\label{ShContEqBis}
K'_n(a_1,\ldots,a_n)=
\left\{
\begin{array}{lr}
(\E-1)K_n(a_1,\ldots,a_n) + 1,
& n\;\hbox{even},\\[6pt]
(\E-1)K_n(a_1,\ldots,a_n),
& n\;\hbox{odd}.
\end{array}
\right.
\end{equation}

In degree $k \geq 2$, the shadow part contains the same monomials
$a_{i_1}\cdots{}a_{i_k}$ as the classical continuant, but with coefficient $(k-1)$.
For example, 
\begin{eqnarray*}
K'_0&=&0,
\\
K'_1(a_1)&=&0,
\\
K'_2(a_1,a_2)&=&a_1a_2,
\\
K'_3(a_1,a_2,a_3)&=&2a_1a_2a_3,
\\
K'_4(a_1,\ldots,a_4)&=&a_1a_2+a_1a_4+a_3a_4+3a_1a_2a_3a_4,
\\
K'_5(a_1,\ldots,a_5)&=&2a_1a_2a_3+2a_1a_2a_5+2a_1a_4a_5+2a_3a_4a_5+
4a_1a_2a_3a_4a_5.
\end{eqnarray*}

\begin{rem}
The shadowed continuant $\cK_n(a_1,\ldots,a_n)$
is a special instance of the supercontinuant
$K\bigl(\begin{smallmatrix} a_1\\
\begin{smallmatrix} \beta_1 & \beta_1
\end{smallmatrix}
\end{smallmatrix}|\ldots |\begin{smallmatrix}
a_n\\\begin{smallmatrix} \beta_n & \beta_n
\end{smallmatrix}
\end{smallmatrix}\bigr)$
of Ustinov \cite{MGOT, Ust}, obtained by substituting
$$
\beta_{2i+1}=(-1)^ia_{2i+1}\xi,
\qquad\qquad
\beta_{2i}=(-1)^{i+1}a_{2i}\eta.
$$
Therefore, Ustinov's determinant formulas and combinatorial algorithm 
can be applied to $\cK_n(a_1,\ldots,a_n)$.
\end{rem}

\subsection{Continued fractions, explicitly}\label{ExpCon}
An analogue of the classical formula~\eqref{ContCFEq}
holds for both even and odd shadowed continued fractions:
\begin{thm}
\label{CFProp}
The supersymmetric continued fraction is
a quotient of shadowed continuants:
\begin{equation}
\label{ShadContCFEq}
\left\{a_1,\ldots,a_n\right\}=
\frac{\cK_n(a_1,\ldots,a_n)}{\cK_{n-1}(a_2,\ldots,a_n)}.
\end{equation}
\end{thm}

\begin{proof}
We will need the following linear recurrence relation.
Its proof is a short exercise combining \eqref{ClassRec}
with the definition of $\cK_n$.

\begin{lem}
\label{RecProp}
One has
\begin{equation}
\label{SuperRec}
\cK_n(a_1, \ldots, a_n) =
\left\{
\begin{array}{lr}
a_1\cK_{n-1}(a_2,\ldots,a_n)(1+\xi\eta)
+ \cK_{n-2}(a_3,\ldots,a_n),
& n\;\mbox{\it even},\\[6pt]
a_1\cK_{n-1}(a_2,\ldots,a_n)(1+\xi\eta)
+ \cK_{n-2}(a_3,\ldots,a_n) - a_1\xi\eta,
& n\;\mbox{\it odd}.
\end{array}
\right.
\end{equation}
\end{lem}

In order to give an alternate formulation of this lemma,
let $\overline{n}$ denote $n$ modulo 2, i.e.,
$$
\overline{n} :=
\Big\{
\begin{array}{l}
\text{0 if $n$ is even,} \\
\text{1 if $n$ is odd.}
\end{array}
$$
The shadow part of \eqref{SuperRec} is equivalent to
\begin{equation}
\label{SuperRecBis}
K'_n(a_1, \ldots, a_n) =
a_1 K_{n-1}(a_2, \ldots, a_n) + a_1 K'_{n-1}(a_2, \ldots, a_n)
+ K'_{n-2}(a_3, \ldots, a_n) - a_1 \overline{n}.
\end{equation}

We will outline the proof of the theorem in the case that $n = 2m$ is even,
leaving the remaining details to the reader.  Consider the matrices
$$
A_k=\begin{pmatrix}
k&1&k\xi\\[2pt]
1&0&0\\[2pt]
k\xi&0&1
\end{pmatrix},
\qquad\qquad
B_\ell=
\begin{pmatrix}
\ell&1&\!\!-\ell\eta\\[2pt]
1&0&0\\[2pt]
\ell\eta&0&1
\end{pmatrix},
\qquad\qquad
\tilde S=
\begin{pmatrix}
0 & 1 & 0 \\[2pt]
1 & 0 & 0 \\[2pt]
0 & 0 & 1
\end{pmatrix}.
$$
Note that $A_k = \cR^k \tilde S$ and $B_\ell = \tilde S \cL^\ell$.
These two sequences of matrices encode the recurrence~\eqref{SuperRec}.
Because $\tilde S^2 = I$, we have $\cR^k \cL^\ell=A_k B_\ell$, and so
$$
\cR^{a_1}\cL^{a_2}\cdots\cR^{a_{2m-1}}\cL^{a_{2m}}=
A_{a_1}B_{a_2}\cdots{}A_{a_{2m-1}}B_{a_{2m}}.
$$
Using \eqref{SuperRec}, one proves by induction that
$A_{a_1}B_{a_2}\cdots{}A_{a_{2m-1}}B_{a_{2m}}$ is equal to
$$
\begin{pmatrix}
\cK_{2m}(a_1,a_2,\ldots,a_{2m}) & 
\cK_{2m-1}(a_1,a_2,\ldots,a_{2m-1}) &
\begin{array}{l}
-(K_{2m}(a_1,a_2,\ldots,a_{2m}) -1)\eta\\[2pt]
+K_{2m-1}(a_1,a_2,\ldots,a_{2m-1})\xi
\end{array}
\\[16pt]
\cK_{2m-1}(a_2,\ldots,a_{2m}) & 
\cK_{2m-2}(a_2,\ldots,a_{2m-1}) &
\begin{array}{l}
(K_{2m-2}(a_2,\ldots,a_{2m-1})-1)\xi\\[2pt]
-K_{2m-1}(a_2,\ldots,a_{2m})\eta
\end{array}
\\[16pt]
\begin{array}{l}
(K_{2m}(a_1,a_2,\ldots,a_{2m}) -1)\xi\\[2pt]
+K_{2m-1}(a_2,\ldots,a_{2m})\eta
\end{array}
&   
\begin{array}{l}
K_{2m-1}(a_1,a_2,\ldots,a_{2m-1})\xi\\[2pt]
+(K_{2m-2}(a_2,\ldots,a_{2m-1})-1)\eta 
\end{array}
& 1- 
\begin{array}{l}
(K_{2m}(a_1,a_2,\ldots,a_{2m})\\[2pt]
+K_{2m-2}(a_2,\ldots,a_{2m-1})-2)\xi\eta
\end{array}
\end{pmatrix}
$$
This leads to the statement of Theorem~\ref{CFProp}.
\end{proof}

\begin{rem}
The coefficient $1+\xi\eta$ appearing in~\eqref{SuperRec}
is reminiscent of the coefficient $1+\xi_i\xi_j$ in the exchange relations
formula in~\cite{OS}; see Eq.~(2).
It would be interesting to establish a direct relation
between~\eqref{SuperRec} and
the version of cluster superalgebras developed in~\cite{OS}.
\end{rem}

\section{Positivity, localization, and convergence}\label{BoundSec}

In this section we apply \eqref{ShadContCFEq} to prove 
various properties of supersymmetric continued fractions.

\subsection{Positivity}\label{PosSec}

This section concerns the positivity of the even and odd shadows
$\big(\frac{p}{q}\big)_{\ES}$ and $\big(\frac{p}{q}\big)_{\OS}$
of non-integral positive rationals $\frac{p}{q}$.
(For the integral case, see~\eqref{int case}.)

\begin{thm}
\label{PosProp}
\begin{enumerate}
\item[(i)]
For $\frac{p}{q} > 1$ and non-integral, both
$\big(\frac{p}{q}\big)_{\ES}$ and $\big(\frac{p}{q}\big)_{\OS}$
are positive rationals.

\smallbreak\item[(ii)]
For $0 < \frac{p}{q} < 1$, both
$\big(\frac{p}{q}\big)_{\ES}$ and $\big(\frac{p}{q}\big)_{\OS}$
are negative rationals.
\end{enumerate}
\end{thm}

\begin{proof}
We will need the following technical lemma.

\begin{lem}
\label{PosLem}
For all positive integers $a_1,\ldots,a_n$, $n \geq 2$,
the following determinant is positive:
$$
D_n(a_1,\ldots,a_n):=\det
\begin{pmatrix}
K'_n(a_1,\ldots,a_n) \phantom{+} & K_n(a_1,\ldots,a_n) \phantom{+} \\[4pt]
K'_{n-1}(a_2,\ldots,a_n) & K_{n-1}(a_2,\ldots,a_n)
\end{pmatrix}.
$$
\end{lem}

\noindent
\textit{Proof of the lemma.}
In this proof, let $K_{n-r}$ denote $K_{n-r}(a_{r+1}, \ldots, a_n)$,
and similarly for $D_{n-r}$.  For $n \geq 2$,
\eqref{ClassRec} and \eqref{SuperRecBis} imply that
\begin{equation} \label{DetRec}
D_n = a_1 K_{n-1} (K_{n-1} - \overline{n}) - D_{n-1},
\end{equation}
where $\overline n$ denotes $n$ modulo~$2$, as earlier.
This gives the alternating sum
$$
D_n = a_1 K_{n-1} (K_{n-1} - \overline{n})
- a_2 K_{n-2} (K_{n-2} + \overline{n} - 1)
+ a_3 K_{n-3} (K_{n-3} - \overline{n}) - \cdots,
$$
the last term in the sum being $D_2$ if $n$ is even and $D_3$ if $n$ is odd.
Note that both $D_2$ and $D_3$ are positive:
$$
D_2 = a_1 a_2^2, \qquad
D_3 = a_2 a_3 \big( a_1 + a_3(a_1 a_2 - 1) \big).
$$

We will show that each pair of consecutive terms in this alternating sum is positive,
proving the lemma.  It suffices to treat the first two terms.  If $n$ is even, they are
$$
a_1 K_{n-1}^2 - a_2 K_{n-2}^2 + a_2 K_{n-2} =
a_1 (a_2 K_{n-2} + K_{n-3})^2 - a_2 K_{n-2}^2 + a_2 K_{n-2}.
$$
The $a_i$ are all positive integers, so this is positive.
On the other hand, if $n$ is odd, the first two terms are
$$
a_1 K_{n-1} (K_{n-1} - 1) - a_2 K_{n-2}^2 =
a_1 (a_2 K_{n-2} + K_{n-3}) (a_2 K_{n-2} + K_{n-3} - 1) - a_2 K_{n-2}^2,
$$
which is again positive, because $n-3$ even
implies that $K_{n-3} - 1$ is non-negative.
\hfill
\qedsymbol

\medbreak
We now prove Theorem~\ref{PosProp}.  For $n \geq 2$,
it follows from~\eqref{ShadContCFEq} that the shadow of
an arbitrary continued fraction $\left[a_1,\ldots,a_n\right]$ is given by
$$
\left[a_1, \ldots, a_n\right]_\rmS =
\frac{D_n(a_1, \ldots, a_n)}{K_{n-1}(a_2, \ldots, a_n)^2}.
$$
Therefore Part~(i) is a direct consequence of Lemma~\ref{PosLem}.
For Part~(ii), observe that $0 < [a_1,\ldots,a_n] < 1$ implies $a_1=0$,
and so \eqref{DetRec} becomes
$D_n(a_1,\ldots,a_n) = -D_{n-1}(a_2,\ldots,a_n)$.
This completes the proof.
\end{proof}

\subsection{Localization}\label{LocSec}

For even shadows, the following strengthening of Theorem~\ref{PosProp}
is immediate from Proposition~\ref{TransProp}.
Note that it does not extend to odd shadows;
for example, $\big(\frac{7}{2}\big)_{\OS}=\frac{5}{4}$.

\begin{cor}
\label{BoundProp}
For every rational $\frac{p}{q}$ in the interval $[n, n+1]$, the even shadow 
$\big(\frac{p}{q}\big)_{\ES}$ is in $[n-1,n]$.
\end{cor}

\subsection{Convergence}\label{ConvSec}

We are now prepared to prove Theorem~\ref{ConvThm}.
Let $a_1,a_2,\ldots$ be a sequence of integers such that
$a_1 \geq 0$ and $a_i\geq1$ for $i>1$, and let
$$
X_n = \{a_1, a_2, \ldots, a_n \} =
[a_1, a_2, \ldots, a_n] + \xi\eta [a_1,a_2,\ldots,a_n]_\rmS
$$
be the $n^{\thup}$ convergent of the corresponding supersymmetric continued fraction.
Write $X_n = x_n + x'_n \xi\eta$,
where $x_n = [a_1, a_2, \ldots, a_n]$ is the classical continued fraction
and $x'_n = [a_1, a_2, \ldots, a_n]_\rmS$ is its shadow.
We must prove that the sequence $(x'_n)$ converges.
Recall the Euler vector field $\E$ defined in~\eqref{EulerEq}.

\begin{lem}
\label{ShadConvLem}
One has
$$
\left[ a_1, a_2, \ldots, a_n \right]_\rmS =
\E \left( \frac{K_n(a_1, \ldots, a_n)} {K_{n-1}(a_2, \ldots, a_n)} \right)
+ \frac{1- \overline{n}} {K_{n-1}(a_2, \ldots, a_n)}
- \frac{ \overline{n}\, K_n(a_1, \ldots, a_n)} {K_{n-1}(a_2, \ldots, a_n)^2}.
$$
\end{lem}

\begin{proof}
By \eqref{ShadContCFEq} and~\eqref{ShContDefBis}, 
$$
\left[a_1, \ldots, a_n \right]_\rmS =
\frac{K'_n(a_1, \ldots, a_n)} {K_{n-1}(a_2, \ldots, a_n)}
- \frac{K_n(a_1, \ldots, a_n) K'_{n-1}(a_2, \ldots, a_n)} {K_{n-1}(a_2, \ldots, a_n)^2}.
$$
From here, compute using \eqref{ShContEqBis} and the quotient rule.
\end{proof}

Now consider $x'_n - x'_{n-1}$.
It is well-known that the classical continuants satisfy the identity
$$
K_n(a_1, \ldots, a_n) K_{n-2}(a_2, \ldots, a_{n-1}) -
K_{n-1}(a_1, \ldots, a_{n-1}) K_{n-1}(a_2, \ldots, a_n)
= (-1)^n.
$$
Therefore the difference of the two terms involving $\E$ simplifies:
\begin{equation}
\label{AlterEqo}
\E \left( \frac{ K_n(a_1, \ldots, a_n)} {K_{n-1}(a_2, \ldots, a_n)}
- \frac{ K_{n-1}(a_1, \ldots, a_{n-1})} {K_{n-2}(a_2, \ldots, a_{n-1})}
\right)
=
\E\left(
\frac{(-1)^n}
{K_{n-1}(a_2, \ldots, a_n) K_{n-2}(a_2, \ldots, a_{n-1})}
\right).
\end{equation}

\begin{lem}
\label{LastLem}
There is a constant $C$ such that
$|x'_{n} - x'_{n-1}| \leq C (a_1 + 1) \varphi^{-n}$,
where $\varphi$ is the golden ratio.
\end{lem}

\begin{proof}
First let us prove that there is a constant $C$ such that
\eqref{AlterEqo}~$\leq C n \varphi^{-2n}$.  Rewrite it as
$$
(-1)^{n-1} \frac{\E \big( K_{n-1}(a_2, \ldots, a_n) K_{n-2}(a_2, \ldots, a_{n-1}) \big)}
{\big( K_{n-1}(a_2, \ldots, a_n) K_{n-2}(a_2, \ldots, a_{n-1}) \big)^2}.
$$
Recall that $a_2, \ldots, a_n$ are positive integers.
Because $K_n$ is of degree~$n$,
$$
\E\big( K_{n-1}(a_2, \ldots, a_n) K_{n-2}(a_2, \ldots, a_{n-1}) \big) \leq
(2n-3) K_{n-1}(a_2, \ldots, a_n) K_{n-2}(a_2, \ldots, a_{n-1}).
$$

Let $F_n$ be the Fibonacci sequence, beginning from $F_0 = 1$.
Elementary arguments show that $K_n$
is a sum of $F_n$ monomials.  Therefore
$$
K_{n-1}(a_2, \ldots, a_n) K_{n-2}(a_2, \ldots, a_{n-1}) \geq F_{n-1} F_{n-2},
$$
giving the stated bound for~\eqref{AlterEqo}.

Next, note that the term in $x'_n$ with denominator
$K(a_2, \ldots, a_n)$ is of order $\varphi^{-n}$.
Finally, apply \eqref{ClassRec} to deduce that
the term in $x'_n$ with denominator $K(a_2, \ldots, a_n)^2$
and numerator containing $K(a_1, \ldots, a_n)$ is of order
$(a_1 + 1) \varphi^{-n}$.  The lemma follows.
(We add $1$ to $a_1$ here
because $a_1$ may be~$0$.)
\end{proof}

Theorem~\ref{ConvThm} follows from this lemma and Theorem~\ref{PosProp}.

\begin{rem}
The same argument used in the classical case shows that the difference
$x'_{n}-x'_{n-2}$ is positive if $n$ is odd and negative if $n$ is even.
\end{rem}

\subsection{Concrete examples}\label{ExaSec}

Here we give numerical approximations of two shadows of irrationals.

\begin{ex}
\label{GoldProp}
The simplest infinite continued fraction is the golden ratio,
$$
\varphi=\frac{1+\sqrt{5}}{2}=[1,1,1,\ldots].
$$
Computation suggests that its shadow is
$$
\left(\varphi\right)_\rmS=\frac{5+\sqrt{5}}{10}=\frac{1}{1+\varphi}.
$$
Compare this prediction to Corollary~\ref{PhiCorol}.
\end{ex}

The shadows of the convergents are presented in Fig.~\ref{fiGold}.
The sequence converges quickly and apparently monotonically,
in contrast with the (rather fractal) figures \ref{Sha12Fig} and~\ref{OddSha12}.
\begin{figure}[htpb]
    \centering
    \includegraphics[width=0.8\textwidth]{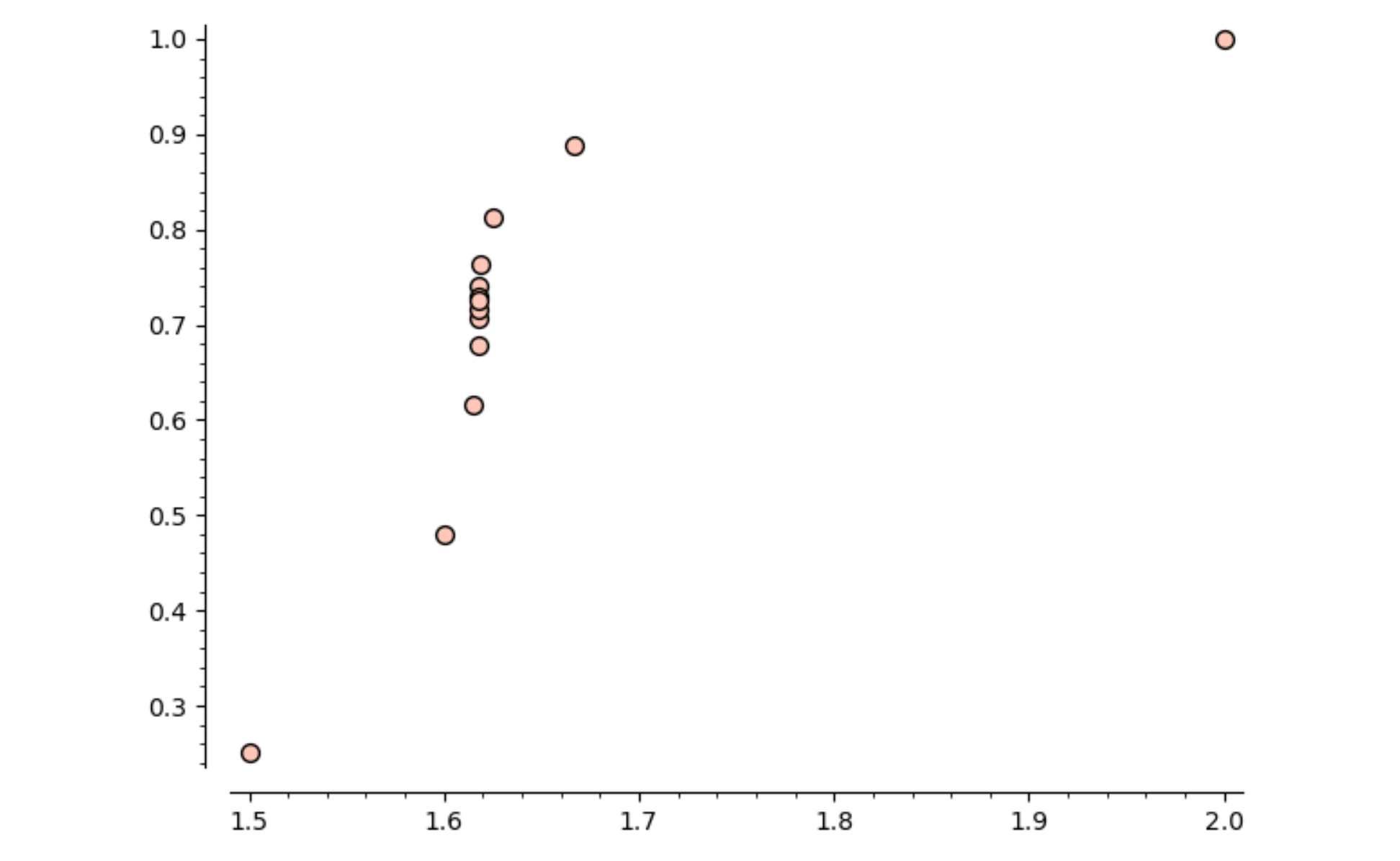}
    \caption{Golden ratio shadow convergence.}
    \label{fiGold}
\end{figure}

\begin{ex}
The second most simple example is the number
$$
\d = 1+\sqrt{2} = [2,2,2,\ldots],
$$
often called the silver ratio.
Computation suggests that its shadow is
$$
\left(\d\right)_\rmS=1+\frac{\sqrt{2}}{2}=\frac{1+\d}{2}.
$$
\end{ex}

\section{The shadowed Farey tree}\label{Far}

In this section we define a variant of the Farey tree of rationals
which contains Grassmann variables and produces shadows.

\subsection{The classical Farey tree}

The Farey tree can be constructed inductively,
by iterating the ``local branching rule''
on the following initial ``fishbone diagram'':
$$
\xymatrix @!0 @R=0.38cm @C=0.38cm
{
&&{\frac{1}{-1}}&&&&&&&&&&&&&{\frac{p-r}{q-s}}&&\\
&&&&&&&&&&&&&&&&&\\
&&\bullet\ar@{-}[dd]\ar@{-}[rru]\ar@{-}[llu]&&&&&&&&&&&&&\bullet\ar@{-}[rru]\ar@{-}[llu]\ar@{-}[dd]\\
&{\frac{1}{0}}&&&{\frac{0}{1}=\frac{0}{-1}}&&&&&&&&&&{\frac{p}{q}}&&{\frac{r}{s}}\\
&&\bullet\ar@{-}[lld]\ar@{-}[rrd]&&&&&&&&&&&&&\bullet\ar@{-}[lld]\ar@{-}[rrd]\\
&&&&&&&&&&&&&&&&&\\
&&{\frac{1}{1}}&&&&&&&&&&&&&{\frac{p+r}{q+s}}&&
}
$$
The quantity $\frac{p}{q}*\frac{r}{s}=\frac{p+r}{q+s}$
is called the \textit{mediant}, or \textit{Farey sum}.

We present here the Farey tree growing downward from the initial diagram.
The tree is ``doubly infinite'': the part growing upward contains the negative rationals.
$$
\begin{small}
\xymatrix @!0 @R=0.38cm @C=0.38cm
{&&&&&&&&&&&&&&&&&&&&&&&&&&\\
&&&&&&&&&&&{-\frac{2}{1}}&&\bullet\ar@{-}[lld]\ar@{-}[lu]&&&{-\frac{1}{1}}
&&&\bullet\ar@{-}[ru]\ar@{-}[rrd]&&{-\frac{1}{2}}\\
&&&&&&&&&&&&&&&&\bullet\ar@{-}[dd]\ar@{-}[rrru]\ar@{-}[lllu]&&&&&&&\\
&&&&&&&&&&&&&&{\frac{1}{0}}&&&&{\frac{0}{1}}&&&&&&&\\
&&&&&&&&&&&&&&&&\bullet\ar@{-}[lllllllldd]\ar@{-}[rrrrrrrrdd]&&&&&&&&\\
&&&&&&&&&&&&&&&&{\frac{1}{1}}\\
&&&&&{\frac{1}{0}}&&&\bullet\ar@{-}[lllldd]\ar@{-}[rrrrdd]
&&&{\frac{1}{1}}&&&&&&&&&&{\frac{1}{1}}&&&\bullet\ar@{-}[lllldd]\ar@{-}[rrrrdd]&&&{\frac{0}{1}}\\
&&&&&&&&{\frac{2}{1}}
&&&&&&&&&&&&&&&&{\frac{1}{2}}\\
&&&&\bullet\ar@{-}[lldd]_{\frac{1}{0}}\ar@{-}[rrdd]^{\frac{2}{1}}
&&&&&&&&\bullet\ar@{-}[lldd]_{\frac{2}{1}}\ar@{-}[rrdd]^{\frac{1}{1}}
&&&&&&&&\bullet\ar@{-}[lldd]_{\frac{1}{1}}\ar@{-}[rrdd]^{\frac{1}{2}}
&&&&&&&&\bullet\ar@{-}[lldd]_{\frac{1}{2}}\ar@{-}[rrdd]^{\frac{0}{1}}\\
&&&&&&&&&&&&&&&&&&&&&&&&&&&&\\
&&\bullet\ar@{-}[ldd]_{\frac{1}{0}}\ar@{-}[rdd]
&&{\frac{3}{1}}&&\bullet\ar@{-}[ldd]\ar@{-}[rdd]^{\frac{2}{1}}
&&&&\bullet\ar@{-}[ldd]_{\frac{2}{1}}\ar@{-}[rdd]
&&{\frac{3}{2}}&&\bullet\ar@{-}[ldd]\ar@{-}[rdd]^{\frac{1}{1}}
&&&&\bullet\ar@{-}[ldd]_{\frac{1}{1}}\ar@{-}[rdd]
&&{\frac{2}{3}}&&\bullet\ar@{-}[ldd]\ar@{-}[rdd]^{\frac{1}{2}}
&&&&\bullet\ar@{-}[ldd]_{\frac{1}{2}}\ar@{-}[rdd]
&&{\frac{1}{3}}&&\bullet\ar@{-}[ldd]\ar@{-}[rdd]^{\frac{0}{1}}\\
&&&&&&&&&&&&&&&&&&&&&&
&&&&&&&&&&&\\
&&{\frac{4}{1}}
&&&&{\frac{5}{2}}
&&&&{\frac{5}{3}}
&&&&{\frac{4}{3}}
&&&&{\frac{3}{4}}
&&&&{\frac{3}{5}}
&&&&{\frac{2}{5}}
&&&&{\frac{1}{4}}&&\\
&&&&\ldots&&&&&&&&&&&&\ldots&&&&&&&&&&&&\ldots
}
\end{small}$$

\subsection{The $\PSL(2,\Z)$ action}

The classical Farey tree has a beautiful symmetry under $\PSL(2,\Z)$.
Corresponding to each vertex of the tree is an element of the group of order~$3$.
For example, in the initial diagram we have
$$
\xymatrix @!0 @R=0.3cm @C=0.3cm
{
&\ar@{-}[rrdd]&&&&\frac{0}{1}&&&&\\
&&&&&\\
-\frac{1}{1}&&&U\ar@{-}[lldd]\ar@{-}[rrrr]&&&&
V\ar@{-}[rruu]\ar@{-}[rrdd]&&&\frac{1}{1}\ \ ,\\
&&&&&\\
&&&&&\frac{1}{0}&&&&
}
$$
where (using the symbol ``$\equiv$'' as a reminder that
elements of $\PSL(2, \Z)$ are cosets)
$$
U \equiv \begin{pmatrix}
0&-1\\[2pt]
1&1
\end{pmatrix},
\qquad\qquad
V \equiv \begin{pmatrix}
1&-1\\[2pt]
1&0
\end{pmatrix}.
$$
The linear fractional transformations by which $U$ and $V$ act
rotate the Farey graph by $2\pi/3$ counterclockwise around their vertices:
$$
\xymatrix
{
\frac{1}{0}\ar@{->}[r]^U&\frac{0}{1}\ar@{->}[r]^U&-\frac{1}{1}\ar@{->}[r]^U&\frac{1}{0}
},
\qquad\qquad
\xymatrix
{
\frac{1}{0}\ar@{->}[r]^V&\frac{1}{1}\ar@{->}[r]^U&\frac{1}{0}\ar@{->}[r]^U&-\frac{1}{0}=\frac{1}{0}.
}
$$

Similarly, corresponding to each edge of the tree is an element of order~$2$.
In the initial diagram,
$$
\xymatrix @!0 @R=0.3cm @C=0.3cm
{
&\ar@{-}[rrdd]&&&&\frac{0}{1}&&&&\\
&&&&&\\
-\frac{1}{1}&&&\bullet\ar@{-}[lldd]\ar@{-}[rrrr]|-{S}&&&&
\bullet\ar@{-}[rruu]\ar@{-}[rrdd]&&&\frac{1}{1}\\
&&&&&\\
&&&&&\frac{1}{0}&&&&
}
$$
where
$$
S \equiv \begin{pmatrix}
0&-1\\[2pt]
1&0
\end{pmatrix}.
$$
The linear fractional transformation by which $S$ acts
rotates the Farey tree around its edge by $\pi$.

Note that conjugating $U$ by $S$ gives $V$.
The elements corresponding to all vertices and edges
of the tree may be obtained from $U$ and $S$
by repeated conjugation:
$$
\xymatrix @!0 @R=0.7cm @C=0.5cm
{
&UVU^{-1}\ar@{-}[rrdd]|-{USU^{-1}}&&&&&&&&V^{-1}UV\\
&&&&&\\
&&&U\ar@{-}[lldd]|-{U^{-1}SU}\ar@{-}[rrrr]|-{S}&&&&
V\ar@{-}[rruu]|-{V^{-1}SV}\ar@{-}[rrdd]|-{VSV^{-1}}&&&\\
&&&&&\\
&U^{-1}VU&&&&&&&&VUV^{-1}
}
$$
In fact, the vertices and edges of the tree are in bijection with the
order~$3$ and order~$2$ subgroups of $\PSL(2, \Z)$, respectively.

\subsection{A Farey tree with Grassmann variables}

Recall from~\eqref{UVEq} the elements $\cU$ and $\cV$ of $\OSp(1|2, \Z)$.
As in the classical case, they are of order~$3$, and
they cycle the vectors around the vertices of the
initial fishbone diagram~\eqref{SuperFish}.
We depict this as follows:
$$
\xymatrix @!0 @R=0.38cm @C=0.9cm
{
&\mkern-14mu -1&&\\
&\mkern-14mu \phantom{-}1&&\\
&\mkern-14mu \phantom{-}\eta&&\\
&&&\\
1&\cU \ar@{-}[ruu]\ar@{-}[luu]\ar@{-}[dd]&0\\
0&&1\\
0&\cV \ar@{-}[ldd]\ar@{-}[rdd]&0\\
&&&\\
&1&&\\
&1&&\\
&\xi&&
}
$$
Repeated conjugation extends the tree and associates
an order~$3$ subgroup of $\OSp(1|2, \Z)$ to each vertex:
$$
\xymatrix @!0 @R=0.38cm @C=0.9cm
{
\cdots&&&&\cdots\\
&&&&\\
&\cU^{-1}\cV\cU\ar@{-}[luu]&&\cU\cV\cU^{-1}\ar@{-}[ruu]&\\
&&&&\\
&&\cU\ar@{-}[ruu]\ar@{-}[luu]\ar@{-}[dd]&\\
&&&\\
&&\cV\ar@{-}[ldd]\ar@{-}[rdd]&\\
&&&&\\
&\cV\cU\cV^{-1}\ar@{-}[ldd]&&\cV^{-1}\cU\cV\ar@{-}[rdd]&\\
&&&&
\\
\cdots&&&&\cdots
}
$$

The edges can also be labelled, in a compatible way.
Recall the element $S$ of $\OSp(1|2, \Z)$
from \eqref{RLUVEq}, and let $\tau$ be the automorphism of the underlying
coordinate ring \eqref{REq} defined by $\eta \mapsto \xi \mapsto -\eta$.
Extend $\tau$ to an outer automorphism of $\OSp(1|2, \Z)$,
acting entry-wise.  Then $\tau$ and $S$ commute, and $\tau S$
is an order~$4$ element with square $-I$ whose action exchanges
$\cU$ and $\cV$ and rotates the initial diagram by~$\pi$.

Applying the group elements to the vectors in the initial diagram
gives the shadowed Farey tree below.
We note that it is possible to construct other Grassmann Farey trees,
for example beginning from the initial vectors in Proposition~\ref{AccordProp}.
$$
\begin{small}
\xymatrix @!0 @R=0.55cm @C=0.52cm
{
&&&&&&&&&&2+\xi\eta&&&&&\!\!-1&&&&&1\\
&&&&&&&&&&-1
&&\bullet\ar@{-}[ldd]\ar@{-}[lu]&&&{1}&&&\bullet\ar@{-}[ru]\ar@{-}[rdd]&&-2-\xi\eta\\
&&&&&&&&&&-\xi-\eta&&&&&\eta&&&&&\xi-\eta\\
&&&&&&&&&&&&&1&&\bullet\ar@{-}[dd]\ar@{-}[rrruu]\ar@{-}[llluu]&&0&&&&&&\\
&&&&&&&&&&&&&0&&&&1\\
&&&&&&&&&&&&&0&&\bullet\ar@{-}[lllllllldd]\ar@{-}[rrrrrrrrdd]&&0&&&&&&\\
&&&&&&&&&&&&&&&1\\
&&&&&&&\bullet\ar@{-}[lllldd]\ar@{-}[rrrrdd]
&&&&&&&&1&
&&&&&&&\bullet\ar@{-}[lllldd]\ar@{-}[rrrrdd]&&&\\
&&&&&&&2+\xi\eta
&&&&&&&&\xi&&&&&&&&1\\
&&&\bullet\ar@{-}[lldd]\ar@{-}[rrdd]
&&&&1&&&&\bullet\ar@{-}[lldd]\ar@{-}[rrdd]
&&&&&&&&\bullet\ar@{-}[lldd]\ar@{-}[rrdd]
&&&&2+\xi\eta&&&&\bullet\ar@{-}[lldd]\ar@{-}[rrdd]\\
&&&3+2\xi\eta&&&&\xi+\eta&&&&3+2\xi\eta&&&&&&&&2+\xi\eta&&&&\xi-\eta&&&&1\\
&\bullet\ar@{-}[ldd]\ar@{-}[rdd]
&&1&&\bullet\ar@{-}[ldd]\ar@{-}[rdd]
&&&&\bullet\ar@{-}[ldd]\ar@{-}[rdd]
&&2+\xi\eta&&\bullet\ar@{-}[ldd]\ar@{-}[rdd]
&&&&\bullet\ar@{-}[ldd]\ar@{-}[rdd]
&&3+2\xi\eta&&\bullet\ar@{-}[ldd]\ar@{-}[rdd]
&&&&\bullet\ar@{-}[ldd]\ar@{-}[rdd]
&&3+2\xi\eta&&\bullet\ar@{-}[ldd]\ar@{-}[rdd]\\
&&&\eta&&&&&&&&3\xi+\eta&&&&&&&&3\xi-\eta&&
&&&&&&-\eta&&&&&\\
&4+2\xi\eta
&&&&5+6\xi\eta
&&&&5+6\xi\eta
&&&&4+2\xi\eta
&&&&3+2\xi\eta
&&&&3+2\xi\eta
&&&&2+\xi\eta
&&&&1&&\\
&1
&&&&2+\xi\eta
&&&&3+2\xi\eta
&&&&3+2\xi\eta
&&&&4+2\xi\eta
&&&&5+6\xi\eta
&&&&5+6\xi\eta
&&&&4+2\xi\eta
\\
&0&&&&\xi+3\eta&&&&4\xi+3\eta&&&&4\xi&&&&4\xi&&&&4\xi-3\eta&&&&\xi-3\eta&&&&0
}
\end{small}
$$

\subsection{Observations}\label{FareyObs}
Our Grassmann Farey tree gives yet another version of the shadow of a rational number,
which should be compared to the even and odd shadows $\ES$ and $\OS$.
We will use the notation $\big(\frac{p}{q}\big)_{\FS}$ for the Farey shadow of~$\frac{p}{q}$.
The definition is as follows: if the vector
$$
\begin{pmatrix}
p + \hat p \xi\eta\\
q + \hat q \xi\eta\\
\cdots
\end{pmatrix}
$$
occurs in the tree, then $\big(\frac{p}{q}\big)_{\FS}$ is
the coefficient of $\xi\eta$ in $(p + \hat p \xi\eta) / (q + \hat q \xi\eta)$:
$$
\Big(\frac{p}{q}\Big)_{\FS} := \frac{p \hat q - \hat p q}{q^2}.
$$
Because the operators $\cU$ and $\cV$ are related
to $\cR$ and $\cL$ via~\eqref{RLUVEq},
the Farey shadow sometimes coincides
with the even or the odd shadow, but not always.

The integers, represented in the classical Farey tree by $\frac{n}{1}$,
appear on its extreme left branch.  Their Farey shadows are
$(1)_{\FS}=0$, $(2)_{\FS}=1$, $(3)_{\FS}=2$, $(4)_{\FS}=2$,
$(5)_{\FS}=2$, $(6)_{\FS}=3$, etc.
This is the OEIS sequence A004524 \cite{OEIS}:
$$
0, 1, 2, 2, 2, 3, 4, 4, 4, 5, 6, 6, 6, 7, \ldots
$$

The ``Fibonacci branch'' of the classical Farey tree is labelled
by quotients of consecutive Fibonacci numbers: $\frac{F_{n+1}}{F_n}$.
It turns out that the Farey shadow sequence $\big(\frac{F_{n+1}}{F_n}\big)_{\FS}$ 
is closely related to the OEIS sequence A054454, which begins with
$\mathrm{(A054454)}_0 = 0$:
$$
0, 1, 2, 6, 12, 26, 50, 97, 180, 332, 600, 1076, \ldots
$$

\begin{prop}
The shadowed Farey tree contains the sequence of $3$-vectors
\begin{equation}
\label{FiboFS}
\begin{pmatrix}
F_{n+1} + \mathrm{(A054454)}_{n-1}\, \xi\eta \\
F_n + \mathrm{(A054454)}_{n-2}\, \xi\eta \\
\cdots
\end{pmatrix}
\ = \
\begin{pmatrix}
1\\
1\\
\xi
\end{pmatrix},\
\begin{pmatrix}
2+\xi\eta\\
1\\
\xi+\eta
\end{pmatrix},\
\begin{pmatrix}
3+2\xi\eta\\
2+\xi\eta\\
3\xi+\eta
\end{pmatrix},\
\begin{pmatrix}
5+6\xi\eta\\
3+2\xi\eta\\
4\xi+3\eta
\end{pmatrix},\
\begin{pmatrix}
8+12\xi\eta\\
5+6\xi\eta\\
8\xi+4\eta
\end{pmatrix}, \cdots.
\end{equation}
\end{prop}

\begin{proof}
Write $\F$ for $\cV \cU^{-1}$, and $u_0$ for the standard basis vector $e_2$.
Then the Fibonacci branch of the Grassmann Farey tree is labelled
by the vectors $u_n := \F^n u_0$.
Define a sequence $a_n$ by writing $a_{n-1}$ for the coefficient of $\xi\eta$
in the first entry $(u_n)_1$ of $u_n$.
Check that then the coefficient of $\xi\eta$ in the second entry, $(u_n)_2$,
is $a_{n-2}$, and that the sequence $a_n$ satisfies the recurrence
$$
a_{n}=a_{n-1}+a_{n-2}+F_{n+1},
$$
which matches the recurrence of A054454.
\end{proof}

It follows from this proposition that the Farey shadows
of the convergents of the golden ratio satisfy
$$
\left(\frac{F_{n+1}}{F_n}\right)_{\FS}=
\left\{
\begin{array}{l}
\big(\frac{F_{n+1}}{F_n}\big)_{\ES},\quad n \;\hbox{even},\\[8pt]
\big(\frac{F_{n+1}}{F_n}\big)_{\OS},\quad n \;\hbox{odd}.
\end{array}
\right.
$$

Proposition~\ref{FiboFS} implies the following statement,
which should be compared with Example~\ref{GoldProp}.
To prove it, simply apply the formula for A054454 given in the OEIS.

\begin{cor}
\label{PhiCorol}
The sequence of rationals $\big(\frac{F_{n+1}}{F_n}\big)_{\FS}$ converges to
$(\varphi)_\rmS=\frac{5+\sqrt{5}}{10}=\frac{1}{1+\varphi}.$
\end{cor}

\begin{rem}
The different approaches of \cite{MOZ}
and \cite{Ovs1} both lead to the sequence A001629,
the self-convolution of the Fibonacci sequence.
It is worth noting that A054454 is related to A001629 by lacunary summation: 
$$
\mathrm{(A054454)}_n =
\mathrm{(A001629)}_n + \mathrm{(A001629)}_{n-2}
+ \mathrm{(A001629)}_{n-4} + \cdots.
$$
\end{rem}

\section{Open problems and conjectures}\label{OpenSec}

In conclusion, we formulate some conjectures and questions.

\subsection{Properties of the shadow function}\label{ContConj}

This subject is unexplored.
Recall from Theorem~\ref{ConvThm} that the shadow
of an irrational number $x$ the limit of the
supersymmetric continued fractions associated to it.
Consider the irrationals, $\mathrm{Irr}\subset\R$.
Is the shadow function,
$$
\mathrm{S}:\mathrm{Irr} \to \R,
$$
continuous?  Based on the convergence property,
as well as on computer experimentation, we conjecture that the answer is yes.

On the other hand, it is easy to check that the shadow functions
$\ES: \Q \to \Q$ and $\OS: \Q \to \Q$ are discontinuous.
For instance, the sequence $\big(2+\frac{1}{n}\big)_{\OS}$ tends to $2$,
but $(2)_{\OS}=0$.

Another challenging question concerns the iteration of the shadow function $\mathrm{S}$.
Given an irrational~$x$, does the sequence $S^n(x)$ converge?

\subsection{The even dominates the odd}\label{AConj}
We conjecture that for every rational $\frac{p}{q}$, one has
$$
\Big(\frac{p}{q}\Big)_{\ES} \geq
\Big(\frac{p}{q}\Big)_{\FS} \geq
\Big(\frac{p}{q}\Big)_{\OS}.
$$
This is supported by computer experiments.
For instance, for every rational in
Figures~\ref{Sha12Fig} and~\ref{OddSha12},
one finds that
$$
\Big(\frac{p}{q}\Big)_{\ES}>\Big(\frac{p}{q}\Big)_{\OS}.
$$
We believe that this strict inequality holds for every~$\frac{p}{q}$.

\subsection{Towards the supersymmetric modular group}\label{IntroOSP}
The supergroup $\OSp(1|2)$ plays the same role in supergeometry
as $\SL(2)$ does classical geometry; see~\cite{Man}.
It consists of $3 \times 3$ matrices over a supercommutative ring $R$,
satisfying the following conditions:
$$
\begin{pmatrix}
a&b&\g\\[2pt]
c&d&\d\\[2pt]
\a&\b&e
\end{pmatrix}
\qquad
\hbox{such that}
\qquad
\left\{
\begin{array}{rcl}
ad-bc&=&1-\a\b,\\[2pt]
e&=&1+\a\b,\\[2pt]
-a\d+c\g&=&\a,\\[2pt]
-b\d+d\g&=&\b.
\end{array}
\right.
$$
Here $a$, $b$, $c$, $d$, and $e$ are even, i.e., elements of $R_{\bar0}$,
and $\a$, $\b$, $\g$, and $\d$ are odd, i.e., elements of $R_{\bar1}$.

Recall that we write $\OSp(1|2, \Z)$ for $\OSp(1|2)$ over
the ring \eqref{REq} of coordinates on $\Z^{1|2}$.
It contains the matrices $\cR,\cL,\cU,\cV$ and $S$,
and one can prove that it is generated by the four matrices
$\cU$, $\cV$, $S$, and
$$
\begin{pmatrix}
1&0&\xi\\[2pt]
0&1&0\\[2pt]
0&\xi&1
\end{pmatrix}.
$$
We hope to investigate it in more detail in a subsequent paper.

\bigbreak \noindent
{\bf Acknowledgements}.
We are grateful to Sophie Morier-Genoud and Alexander Veselov for enlightening discussions.
C.H.C.\ was partially supported by Simons Foundation Collaboration Grant~519533, and
V.O.\ was partially supported by the ANR project PhyMath, ANR-19-CE40-0021.


\begin{thebibliography}{99}

\bibitem{Lic}
A.~Bapat, L.~Becker, A.~M.~Licata,
{\it $q$-deformed rational numbers and the 2-Calabi--Yau category of type $A_2$},
arXiv:2202.07613.

\bibitem{BR1}
J.~Berstel, C.~Reutenauer, 
{\it Noncommutative rational series with applications},
Encyclopedia of Mathematics and its Applications {\bf 137},
Cambridge University Press, Cambridge, 2011. 

\bibitem{Eul} 
L.~Euler, Introductio in analysin infinitorum, Vol.~I, 1748.

\bibitem{Fre} 
D.~Freed, 
Five lectures on supersymmetry,
American Mathematical Society, Providence, 1999. 

\bibitem{Concr}
R.~Graham, D.~Knuth, O.~Patashnik, 
Concrete mathematics: A foundation for computer science,
Addison-Wesley Publishing Company, Reading, 1989.

\bibitem{Hon}
A.~N.~W.~Hone,
{\it Casting light on shadow Somos sequences},
arXiv:2111.10905.

\bibitem{Man}
Yu.~Manin, 
Topics in noncommutative geometry, 
Princeton University Press, Princeton, 1991.

\bibitem{DM}
J.-P.~Michel, C.~Duval,
{\it On the projective geometry of the supercircle: a unified construction of the super cross-ratio and Schwarzian derivative},
Int.\ Math.\ Res.\ Not.\ IMRN {\bf 2008}, no.~14, Art.~ID rnn054. 

\bibitem{SVRe}
S.~Morier-Genoud, V.~Ovsienko,
{\it On q-Deformed Real Numbers},
Experimental Mathematics, {\bf 32} (2022), no.~2, 652--660, arXiv:1908.04365.

\bibitem{MGOT}
S.~Morier-Genoud, V.~Ovsienko,  S.~Tabachnikov,
{\it Introducing supersymmetric frieze patterns and linear difference operators.
With an appendix by Alexey Ustinov},
Math.\ Z.\ {\bf 281} (2015), no. 3-4, 1061--1087.

\bibitem{MOZ} 
G.~Musiker, N.~Ovenhouse, S.~W.~Zhang, 
{\it Double dimer covers on snake graphs from super cluster expansions},
J.~Algebra {\bf 608} (2022), 325--381.

\bibitem{MOZ1} 
G.~Musiker, N.~Ovenhouse, S.~W.~Zhang, 
{\it Matrix Formulae for Decorated Super Teichm\"uller Spaces},
arXiv:2208.13664.

\bibitem{OEIS} 
The On-Line Encyclopedia of Integer Sequences, OEIS Foundation Inc., http://oeis.org.

\bibitem{Ovs} 
V.~Ovsienko, 
{\it A step towards cluster superalgebras},
arXiv:1503.01894.

\bibitem{Ovs1} 
V.~Ovsienko, 
{\it Shadow sequences of integers, from Fibonacci to Markov and back},
to appear in Math.\ Intelligencer, arXiv:2111.02553.

\bibitem{OS} 
V.~Ovsienko, M.~Shapiro,
{\it Cluster algebras with Grassmann variables},
Electron.\ Res.\ Announc.\ Math.\ Sci.\ {\bf 26} (2019), 1--15.

\bibitem{OT} 
V.~Ovsienko, S.~Tabachnikov,
{\it Dual numbers, weighted quivers, and extended Somos and Gale-Robinson sequences},
Algebr.\ Represent.\ Theory {\bf 21} (2018), no.~5, 1119--1132.

\bibitem{Rab} 
J.~Rabin,
{\it Super elliptic curves},
J.~Geom.\ Phys.\ {\bf 15} (1995), no.~3, 252--280.

\bibitem{Ust}
A.~Ustinov,
{\it Supercontinuants},
arXiv:1503.04497.

\bibitem{Ves}
A.~Veselov,
{\it Conway's light on the shadow of Mordell},
arXiv:2208.14184.

\end{thebibliography}
\end{document}